\documentclass[12pt,letterpaper,english]{amsart}
\usepackage{lmodern}
\usepackage{helvet}

\usepackage[T1]{fontenc}
\usepackage[latin9]{inputenc}
\usepackage{mathrsfs}
\usepackage{amsthm}
\usepackage{amstext}
\usepackage{amssymb}
\usepackage{tikz}
\usepackage[justification=centering]{caption}
\usepackage{esint}
\usepackage{graphicx}
\usepackage{cancel}
\usepackage{comment}
\usepackage{enumitem}
\usepackage{amscd,color}
\DeclareFontFamily{OT1}{pzc}{}
\DeclareFontShape{OT1}{pzc}{m}{it}{<-> s * [1.10] pzcmi7t}{}
\DeclareMathAlphabet{\mathpzc}{OT1}{pzc}{m}{it}
\usepackage[bbgreekl]{mathbbol}
\DeclareSymbolFontAlphabet{\mathbb}{AMSb}
\DeclareSymbolFontAlphabet{\mathbbl}{bbold}
\newcommand*\pFqskip{8mu}
\catcode`,\active
\newcommand*\pFq{\begingroup
        \catcode`\,\active
        \def ,{\mskip\pFqskip\relax}%
        \dopFq
}
\catcode`\,12
\def\dopFq#1#2#3#4#5{%
        {}_{#1}F_{#2}\biggl[\genfrac..{0pt}{}{#3}{#4};#5\biggr]%
        \endgroup
}

\makeatletter

\pdfpageheight\paperheight
\pdfpagewidth\paperwidth

\numberwithin{equation}{section}
\numberwithin{figure}{section}
\usepackage{enumitem}		
\theoremstyle{plain}
\newtheorem{thm}{\protect\theoremname}[section]
\theoremstyle{remark}
\newtheorem{rem}[thm]{\protect\remarkname}
\theoremstyle{definition}

\theoremstyle{definition}
\newtheorem{example}[thm]{\protect\examplename}
\theoremstyle{plain}

\theoremstyle{plain}
\newtheorem{prop}[thm]{Proposition}
\theoremstyle{plain}
\newtheorem{lem}[thm]{Lemma}
\theoremstyle{plain}
\newtheorem{conj}[thm]{\protect\conjecturename}

\usepackage{amsfonts}
\usepackage{amssymb}
\usepackage[all]{xy}
\usepackage{array}
\usepackage{lmodern}
\usepackage[T1]{fontenc}
\usepackage{bm}

\def\uo{\underline{0}}

\def\QQ{\mathbb{Q}}
\def\RR{\mathbb{R}}
\def\CC{\mathbb{C}}

\def\ZZ{\mathbb{Z}}
\def\PP{\mathbb{P}}

\def\T{\mathcal{T}}

\def\vf{\varphi}

\def\ay{\mathbf{i}}

\def\ve{\varepsilon}

\def\C{\mathcal{C}}

\def\ua{\underline{a}}

\def\R{\mathcal{R}}
\def\RT{\tilde{\R}}

\theoremstyle{definition}

\theoremstyle{definition}

\theoremstyle{definition}
\newtheorem*{thx}{Acknowledgments}
\theoremstyle{plain}

\makeatother

\usepackage{babel}
  \providecommand{\corollaryname}{Corollary}
  \providecommand{\definitionname}{Definition}
  \providecommand{\remarkname}{Remark}
\providecommand{\theoremname}{Theorem}
 \providecommand{\examplename}{Example}

\providecommand{\conjecturename}{Conjecture}

\begin{document}

\title{Integral Regulators on Mirror Curves With Mass Parameter}

\author[S. Sinha Babu]{Soumya Sinha Babu}
\address{University of Georgia, Department of Mathematics and Statistics, Athens, GA}
\email{soumya@uga.edu}

\begin{abstract}
In 2015, Codesido-Grassi-Marino laid the foundation of a connection between local CY 3-folds and the spectra of operators attached to their mirror curves in the context of Topological String Theory. In a 2024 paper with C. Doran and M. Kerr, the author previously deduced two consequences of this conjecture; one relating zeroes of higher normal function to the spectra of operators of genus one curves, the other connecting integral regulators of $K_2$-classes on mirror curves to dilogarithm values at algebraic arguments. We now show that the latter continues to hold in the presence of a mass parameter, thus expanding the range of the conjecture.
\end{abstract}
\maketitle


\section{Introduction}\label{S4}

We consider the lattice triangle, $\Delta_{m,n}=\text{Conv}((1,0),(0,1),(-m,-n))$ where $m,n\in \mathbb{Z}_{>0}$. A fan on a triangulation of $\{1\}\times \Delta_{m,n}\subset \RR^3$ determines a \textit{local} (toric, hence, non-compact) CY 3-fold, $X$. Corresponding to $X$ one can construct a family of \emph{mirror curves} $\C_{m,n}\subset \CC^*\times\CC^*$, of genus $g$ equal to the number of interior integer points of $\Delta_{m,n}$, cut out by the Laurent polynomial $F_{m,n}(x,y)$ with Newton polygon $\Delta_{m,n}$. We assume that $F_{m,n}(x,y)$ is \textit{tempered}, that is to say it takes the form,
\begin{equation}
\begin{split}
\textstyle F_{m,n}(x,y)&:=x\textstyle+y+x^{-m}y^{-n}+\sum_{j=1}^g a_j x^{m^{(j)}_1}\mspace{-10mu}y^{m^{(j)}_2} \\
&\textstyle+\sum_{\ell=1}^{g_1-1}\binom{g_1}{\ell}x^{1-\ell\frac{m+1}{g_1}}y^{-\ell\frac{n}{g_1}}+\sum_{\ell=1}^{g_2-1}\binom{g_2}{\ell}x^{-\ell\frac{m}{g_2}}y^{1-\ell\frac{n+1}{g_2}}.
\end{split}
\end{equation}

where $g_1:=\gcd(m+1,n)$ and $g_2=\gcd(m,n+1)$. To keep track of moduli, we rename $\C_{m,n}$ and $F_{m,n}$ to $\C^{\ua}_{m,n},F^{\ua}_{m,n}$ respectively. Of particular importance is the so-called \textit{maximal conifold point}, $\hat{\ua}\in(\CC^*)^g$, a point in moduli at which the family degenerates to a rational nodal curve.\newline
We now focus specifically on integral regulators\footnote{For a complete discussion see \cite[\S2.2]{DKSB}.}. Temperedness of $F^{\ua}_{m,n}$ implies that the Milnor symbol $\{-x,-y\}\in K_2(\CC(\C^{\ua}_{m,n}))$ extends to motivic cohomology classes on the compactifications $\overline{\C^{\ua}_{m,n}}\subset \PP_{\Delta}$.  Writing $R\{f,g\}:=\log(f)\tfrac{dg}{g}-2\pi\ay\log(g)\delta_{T_f}$ for the standard regulator current for Milnor $K_2$-symbols ($T_f:=f^{-1}(\RR_{<0})$ the cut in branch of $\log$), there is a symplectic basis $\{\gamma_j,\beta_j\}_{j=1}^g$ of $ H_1(\overline{\C^{\ua}_{m,n}},\ZZ)$ with regulator periods $R_{\gamma_j}:=\int_{\gamma}R\{-x,-y\}|_{\C^{\ua}_{m,n}}\sim -2\pi\ay\log(a_j)$, and the regulator class $\R=R\{-x,-y\}\in H^1(\overline{\C^{\ua}_{m,n}},\CC/\ZZ(2))$ then has a local lift to $H^1(\overline{\C^{\ua}_{m,n}},\CC)$ given by
\begin{equation*}
\RT=\sum_{\ell=1}^g \left( R_{\gamma_{\ell}}\gamma_{\ell}^*+R_{\beta_{\ell}}\beta_{\ell}^*\right),
\end{equation*}
whose Gauss-Manin derivatives are given by,
\begin{equation*}
\textstyle\omega_j:=\nabla_{{\partial}/{\partial R_{\gamma_j}}}\RT=
\dfrac{a_j}{2\pi\ay}\mathrm{Res}_{\C^{\ua}_{m,n}}\dfrac{dx\wedge dy}{x^{2j}yF^{\ua}_{m,n}}.   
\end{equation*}

Denote by $\hat{\mathcal{C}}_{m,n}$ the fiber of the family over the \textit{maximal conifold} point $\underline{\hat{a}}$. It has $g$ nodes $\{\hat{p}_j\}$, and the cycles $\{\hat{\gamma}_j\}_{j=1}^g$ passing through each node generate $H_1(\hat{\C}_{m,n})$; we set $R_{\hat{\gamma}_j}:=\int_{\hat{\gamma}_j}R\{x,y\}$.  Let $\bm{\kappa}=\,_{\underline{\hat{\gamma}}}[\mathrm{Id}]_{\underline{\gamma}(\hat{a})}$ be the change-of-basis matrix.

Recent work of \cite{GHM,Ma,CGM} connects the enumerative geometry of $X$ to the spectral theory of certain operators $\hat{F}$ on $L^2(\RR)$ attached to $\C^{\ua}_{m,n}$. In \cite{DKSB}, Using tools derived from asymptotic Hodge Theory and local mirror symmetry, we were able to recast this relationship in terms of regulator periods, resulting in the following 

\begin{conj}\label{c3}
For the families $\C^{\ua}_{m,n}$ cut out by $F^{\ua}_{m,n}$, the regulator period $R_{\gamma_1}({\ua})$ asymptotic to $-2\pi\ay\log(a_1)$ at the origin has value
\begin{equation}\label{e2.4.9}
\tfrac{1}{2\pi\ay}R_{\gamma_1}(\hat{\ua})\equiv \tfrac{m+n+1}{\pi}D_2(1+\zeta_{m+n+1}) \;\;\;\textup{mod $\QQ(1)$}
\end{equation}
at the maximal conifold point $\hat{\ua}$.
\end{conj}

\begin{thx}
The author is indebted to Dr. M. Kerr and Dr. C. Doran for their help in understanding and settling the difficult parts of this problem. The author also thanks Dr. P. Bousseau for countless discussions on this subject bringing forward deeper insights into the conjecture.
\end{thx}

\newpage

\section{The main result}\label{S4a}

In \cite[\S4]{DKSB}, we proved Conjecture 1.1 for two infinite families of genus-$g$ curves, namely those cut out by $F^{\ua}_{g,g}$ and $F^{\ua}_{2g-1,1}$. Note that neither $\Delta_{g,g}$ nor $\Delta_{2g-1,1}$ contain any interior integer points, which significantly reduces calculations. In present situation we wish to consider genus-$g$ families cut out by $F^{\ua}_{2g,1}$. There is precisely one interior integer point $(0,-g)$, with corresponding moduli $a_{g+1}$. In physics literature such points are called \textit{mass parameters}. However due to the temperedness condition, we have $$a_{g+1}=-2.$$ 
\begin{thm}\label{thm_01}
Conjecture 1.1 holds for the family $\C^{\ua}_{2g,1}$; that is,
\begin{align}
\tfrac{1}{2\pi\mathbf{i}}R_{\gamma_{1}}(\underline{\hat{a}})\underset{\mathbb{Q}(1)}{\equiv}&\dfrac{2g+2}{\pi}D_2(1+\zeta_{2g+2})\label{0.3}
\end{align}
\end{thm}

\begin{rem}
Note that the Milnor symbol $\{x,y\}$ on the curve defined by substituting $-x,-y$ for $x,y$ resp. then multiplying the equation by $-1$, being a pullback, is integrally tempered with the same integral regulator as $\{-x,-y\}$.  The new equation replaces $a_{g+1}$ by $-a_{g+1}$, and also changes the sign of $a_1,a_3,a_5,\dots$ etc.; it is this new equation which we will use going forward. Note also that Conjecture 1.1 is stated in terms of the regulator period asymptotic to $-2\pi\mathbf{i}\log(a_n)$; it is convenient in this section to drop the negative sign and work with one asymptotic to $2\pi\mathbf{i}\log(a_n)$. Thus from now on $$R_{{\gamma}_n}\sim 2\pi \mathbf{i}\log(a_n).$$
\end{rem}

We prove Theorem \ref{thm_01} by essentially adopting the same techniques discussed in \cite[\S4]{DKSB} - monodromies at $\underline{\hat{a}}$ can be calculated using power series coming from classical periods, following which the limiting regulator periods can be found by parameterizing the spectral curve and attaching specially constructed divisors to it. 
\begin{prop}\label{thm_02}
Let $\kappa_j:=\mathrm{gcd}(j,g+1)$. Then
\begin{equation}
    \bm{\kappa}=\mathrm{diag}(\kappa_1,\dots,\kappa_g).
\end{equation}
It then follows from temperedness that
\begin{equation}
    \tfrac{1}{2\pi\mathbf{i}}R_{\gamma_{j}}(\underline{\hat{a}})\underset{\mathbb{Q}(1)}{\equiv}\tfrac{\kappa_j}{2\pi \mathbf{i}}R_{\hat{\gamma}_{j}}.
\end{equation}
\end{prop}

\subsection{Preliminary results}

It is necessary to understand growth of the power series corresponding to regulator periods and singularities of $\C^{\hat{\ua}}_{2g,1}$.

\begin{lem}\label{lemma_01}
The following identity holds,
\begin{equation}
   \sum_{k=0}^{m/2}\dfrac{2^{-2k}}{\Gamma(1+k)^2\Gamma(1+m-2k)}=\dfrac{\Gamma\bigg(\dfrac{1+2m}{2}\bigg)}{\Gamma(1+m)\Gamma\bigg(\dfrac{2+m}{2}\bigg)\Gamma\bigg(\dfrac{1+m}{2}\bigg)}.
\end{equation}
\end{lem}
 
\begin{proof}
We reduce the given series into a hypergeometric series, and apply \textbf{Gauss' summation theorem} as follows,

\begin{align}
\sum_{k=0}^{m/2}\dfrac{2^{-2k}}{\Gamma(1+k)^2\Gamma(1+m-2k)}&=\dfrac{1}{m!}~\pFq{2}{1}{\dfrac{1-m}{2},-\dfrac{m}{2}}{1}{1}\nonumber\\
&=\dfrac{\Gamma\bigg(\dfrac{1+2m}{2}\bigg)}{\Gamma(1+m)\Gamma\bigg(\dfrac{2+m}{2}\bigg)\Gamma\bigg(\dfrac{1+m}{2}\bigg)}.
\end{align}
\end{proof}

\begin{lem}\label{lemma_02}
If $a,b,c\in\mathbb{R}_{>>>0}$ are such that $a=b+c$, then 

\begin{equation}
\underbrace{\dfrac{2^c\Gamma(1+a)\Gamma\bigg(\dfrac{1+2c}{2}\bigg)}{\Gamma(1+b)\Gamma(1+c)\Gamma\bigg(\dfrac{1+c}{2}\bigg)\Gamma\bigg(\dfrac{2+c}{2}\bigg)}}_{=:\mathcal{A}_{a,b,c}}\approx \dfrac{1}{\sqrt{2}\pi c}\sqrt{\dfrac{a}{b}}\Bigg(\dfrac{a}{b}\bigg(\dfrac{4b}{c}\bigg)^{c/a}\Bigg)^a.
\end{equation}
\end{lem}
 
\begin{proof}
Using \textbf{Duplication formula},

\begin{equation}
\dfrac{1}{\Gamma\bigg(\dfrac{1+c}{2}\bigg)\Gamma\bigg(\dfrac{2+c}{2}\bigg)}=\dfrac{2^c}{\sqrt{\pi}\Gamma(1+c)}    
\end{equation}

Thus

\begin{equation}
\dfrac{\Gamma\bigg(\dfrac{1+2c}{2}\bigg)}{\Gamma\bigg(\dfrac{1+c}{2}\bigg)\Gamma\bigg(\dfrac{2+c}{2}\bigg)}=\dfrac{2^c\Gamma\bigg(c+\dfrac{1}{2}\bigg)}{\Gamma(c+1)}\approx 2^c\sqrt{\dfrac{1}{c}} .   
\end{equation}

wherein we have used a modified \textbf{Sterling's approximation} which says that for large $x\in\mathbb{R}_{\geq 0}$ and $\alpha,\beta\in\mathbb{R}_{>0}$,

\begin{equation}
 \dfrac{\Gamma(x+\alpha)}{\Gamma(x+\beta)}\approx x^{\alpha-\beta}.
\end{equation}

It follows that

\begin{align}
\mathcal{A}_{a,b,c}\approx \dfrac{4^c\Gamma(1+a)}{\sqrt{\pi}\Gamma(1+b)\Gamma(1+c)\sqrt{c}}&\approx \dfrac{1}{\sqrt{2}\pi c}\sqrt{\dfrac{a}{b}}\dfrac{4^c a^a}{b^bc^c}e^{-a+b+c}\nonumber   \\
&=\dfrac{1}{\sqrt{2}\pi c}\sqrt{\dfrac{a}{b}}\dfrac{4^ca^a}{b^{a-c}c^c}\\
&=\dfrac{1}{\sqrt{2}\pi c}\sqrt{\dfrac{a}{b}}\Bigg(\dfrac{a}{b}\bigg(\dfrac{4b}{c}\bigg)^{c/a}\Bigg)^a\nonumber
\end{align}
as was to be shown.

\end{proof}

\begin{lem}\label{Proposition_0.2}
Suppose that the fiber over $\underline{\tilde{a}}=(\tilde{a}_1,\dots,\tilde{a}_{g+1})$ has $g$-many singularities, say $\tilde{p}_j:=(\tilde{x}_j,\tilde{y}_j), n=1,\dots,g$. Then for each $j$, $\tilde{p}_j$ is a node.
\end{lem}

\begin{proof} We argue in the same vein as [DKSB]. The result becomes immediate modulo the hessian calculation, which in this case boils down to the following - we begin by defining
\begin{equation}
 \textstyle   \tilde{P}(x):=g+1+\sum_{j=1}^{g}(g+1-j)\tilde{a}_{j}x^{-j},
\end{equation}
and observing that
\begin{equation}
  \textstyle \tilde{P}(\tilde{p}_j)=\tfrac{g}{\tilde{x}_j}F^{\underline{\tilde{a}}}_{2g,1}(\tilde{p}_j)+\partial_x F^{\underline{\tilde{a}}}_{2g,1}(\tilde{p}_j)=0.
\end{equation}
Thus $\textbf{Z}(\tilde{P})=\{\tilde{p}_1,\ldots,\tilde{p}_{g}\}$, i.e., $\tilde{P}$ has no repeated roots; that is, $\tilde{P}'(\tilde{p}_j)\neq 0$ ($\forall j$).  To compute the Hessians, write
\begin{align}
\textstyle\partial_{xx}F^{\underline{\tilde{a}}}_{2g,1}(\tilde{p}_j)&=\textstyle\sum_{\ell=1}^{g+1}\ell(\ell-1)\tilde{a}_{\ell}\tilde{x}_j^{-\ell-1}+2g(2g+1)\tilde{x}_j^{-2g-2}\tilde{y}^{-1}_{j}\nonumber\\
&\textstyle=\sum_{\ell=1}^{g+1}\ell(\ell-1)\tilde{a}_{\ell}\tilde{x}_{j}^{-\ell-1}+\tfrac{2g(2g+1)\tilde{y}_j}{\tilde{x}_j^2},\\
\partial_{xy}F^{\underline{\tilde{a}}}_{2g,1}(\tilde{p}_j)&=2g\tilde{x}_j^{-2g-1}\tilde{y}_j^{-2}=\tfrac{2g}{\tilde{y}_j},\;\text{and}\\
\partial_{yy}F^{\underline{\tilde{a}}}_{2g,1}(\tilde{p}_j)&=2\tilde{x}_j^{2g}\tilde{y}_j^{-3}=\tfrac{2}{\tilde{y}_j}.
\end{align}
It can be shown that
\begin{equation}
     \partial_{xx}F^{\underline{\tilde{a}}}_{2g,1}(\tilde{p}_j)=\tfrac{2g^2\tilde{y}_j}{2\tilde{x}_j^2}+\tfrac{\tilde{P}'(\tilde{x}_j)}{2},
\end{equation}
therefore,
\begin{align}
H_{F^{\underline{\tilde{a}}}_{2g,1}}(\tilde{p}_j)&=\left(\partial_{xy}F^{\underline{\tilde{a}}}_{2g,1}(\tilde{p}_j)\right)^2-\partial_{xx}F^{\underline{\tilde{a}}}_{2g,1}(\tilde{p}_j)\partial_{yy}F^{\underline{\tilde{a}}}_{2g,1}(\tilde{p}_j)\nonumber\\
  &\textstyle= \tfrac{4g^2}{\tilde{x}_j^2}-\tfrac{4g^2}{\tilde{x}_j^2}-\tfrac{\tilde{P}'(\tilde{x}_j)}{\tilde{y}_j}\; =\;-\tfrac{\tilde{P}'(\tilde{x}_j)}{\tilde{y}_j}\neq 0\nonumber
\end{align}
as was to be shown.
\end{proof}

\subsection{Monodromy calculations via power series}\label{S4b}

Let us recall the key result of [DKSB] for monodromy calculations - consider a $1$-parameter family of curves $\C\to \PP^1$ with coordinate $t$, endowed with a section $\omega$ of the relative dualizing sheaf; on smooth fibers $\C_t$, $\omega_1$ is a holomorphic $1$-form.  Assume that $\C_c$ has a single node $p_c$ (i.e. is a ``conifold fiber''), and let $\delta_0$ be the ``conifold'' vanishing cycle pinched at $p_c$.  Writing $\ve_0$ for a cycle invariant about $t=0$, its monodromy about $t=c$ is a multiple of $\delta_0$, say $\bm{k}\delta_0$ for some $\bm{k}\in\mathbb{Z}_{\geq 0}$. 
\begin{lem}\label{thm_04}
The conifold multiple $\bm{k}$ is computed by
\begin{equation}
\bm{k}=\dfrac{\underset{m\to\infty}{\lim} b_m \cdot c^m\cdot m}{\mathrm{Res}_{p_c}~\omega_c}.\label{0.7}    
\end{equation}
\end{lem}

For the proof of Proposition \ref{thm_02}, we need to compute the Picard-Lefschetz matrix $\bm{\kappa}$, whose entries $\bm{\kappa}_{ij}$ tell how many times the specialization $\gamma_i(\hat{\ua})$ passes through $\hat{p}_j$.  In order to invoke Lemma \ref{thm_04} for this purpose, we should reinterpret these numbers as (roughly speaking) conifold multiples for 1-parameter subfamilies of $\C_{\ua}$ acquiring a \emph{single} node.  The idea is that $\hat{\ua}$ is a normal-crossing point of the discriminant locus, whose $g$ local-analytic irreducible components each parametrize fibers carrying a single node $p_j$.  These are labeled in such a way that the $j^{\text{th}}$ component can be followed out to where it meets the $a_j$-axis at $a_j=\mathring{a}_j$.  Call this fiber $\C_{2g,1}^{\mathring{\ua}_j}$, and $\mathring{p}_j=(\mathring{x}_j,\mathring{y}_j)$ for the limit of the node to it.

We have the 1-forms
\begin{equation}
\varpi_j=\tfrac{1}{2\pi\ay}\nabla_{\delta_{a_j}}R\{x,y\}=\frac{-a_j}{2\pi\ay}\mathrm{Res}_{\C_{2g,1}}\left(\frac{dx\wedge dy}{x^{2j} y F_{2g,1}(x,y)}\right)
\end{equation}
and 1-cycles $\gamma_j$ ($j=1,\ldots,g$).  The computation that follows will consider periods $\Pi_{jj}=\int_{\gamma_j}\varpi_j$ on the 1-parameter families over the $a_j$-axes (acquiring a single node at $a_j =\mathring{a}_j$), which will suffice to determine the diagonal terms $\bm{\kappa}_{jj}$.  That the remaining, off-diagonal terms are actually zero follows from the fact that each $\gamma_j$ is well-defined on a tubular neighborhood of the hyperplane in (compactified) moduli defined by $z_j=0$, which is cut by the conifold components carrying $p_i$ for every $i\neq j$.

Now $\C_{2g,1}^{\mathring{\ua}_j}$ is defined by
\begin{equation}
     f^{(j)}_{2g,1}:=F^{\underline{\mathring{a}}_j}_{2g,1}(x,y)=x+y+\mathring{a}_jx^{1-j}+a_{g+1}x^{-g}+x^{-2g}y^{-1},
\end{equation}
and to find the node we set $f_{2g,1}^{(j)}(\mathring{x}_j,\mathring{y}_j)=\mathring{x}_j\partial_x f_{2g,1}^{(j)}(\mathring{x}_j,\mathring{y}_j)$ which gives rise to equations of the form,
\begin{align}
2\mathring{y}_j+\mathring{x}_j+\mathring{a}_j\mathring{x_j}^{1-j}+a_{g+1}\mathring{x}_j^{-g}&=0,   \\
\mathring{x}_j+(1-j)\mathring{x}_{j}^{-j}-g\mathring{a}_{j}\mathring{x}_{j}^{-g-1}-2g\mathring{x}_j^{-2g-1}\mathring{y}_j^{-1}&=0.
\end{align}
This yields
\begin{align}
    \mathring{x}_{j}&=\sqrt[g+1]{\dfrac{4(g-j+1)}{j}},\\
    \mathring{a}_{j}&=-\dfrac{g+1}{g-j+1}\bigg(\dfrac{4(g-j+1)}{j}\bigg)^{\tfrac{j}{g+1}}.
\end{align}
In particular, we have the relation
\begin{equation}
    \mathring{a}_{j}\mathring{x}_{j}^{g-j+1}=-\dfrac{4(g+1)}{j}.
\end{equation}
In order to calculate the residue of $\varpi_j$ at $\mathring{p}_j$, recall that for any $f(x,y)=Ax^2+Bxy+Cy^2+\text{higher order terms}\in\mathbb{C}[x,y]$, we have
\begin{equation}
 \mathrm{Res}^2_{\uo}\frac{dx\wedge dy}{f} := {\mathrm{Res}}_{\uo}\left(\mathrm{Res}_{\bm{Z}(f)}\dfrac{dx\wedge dy}{f}\right)=\dfrac{1}{\sqrt{B^2-4AC}}.
\end{equation}
Changing variables to $X:=x-\mathring{x}_{j}$, $Y:=y-\mathring{x}_{j}$ in $f^{(j)}_{2g,1}(x,y)$ leads to the equation
\small\begin{align}
x^{2g}yf^{(j)}_{2g,1}=\tfrac{(6g^2-4(j+1)-4g(j-1)}{\mathring{x_j}^2}X^2-2g\mathring{x}_{j}^{g-1}XY 
&+\mathring{x}_{j}^{2g}Y^2 \nonumber\\ &+\text{higher order terms}.
\end{align}\normalsize
Therefore
\begin{align}
   {\mathrm{Res}}^2_{\mathring{p}_{j}}\dfrac{dx \wedge dy}{x^{2g}yf^{(j)}_{2g,1}}&=\tfrac{1}{\mathring{x}_{j}^{g-1}\sqrt{4g^2-2\big(6g^2-4(j+1)-4g(j-1)\big)^2}}\nonumber\\
   &=\tfrac{(-1)^{g+1}}{\ay \mathring{x}_{j}^{g-1}\sqrt{8(g-j+1)(g+1)}}
\end{align}
Consequently the residue of $\varpi_{j}$ may now be found:
\begin{align}
\mathrm{Res}_{\mathring{p}_{j}}\varpi_{j}&=\dfrac{-\mathring{a}_{j}}{2\pi \ay} {\mathrm{Res}}^2_{\mathring{p}_{j}}\dfrac{dx \wedge dy}{x^{2j}yf^{(j)}_{2g,1}}\nonumber\\
&=\dfrac{-\mathring{a}_{j}}{2\pi \ay}\cdot \mathring{x}_{j}^{2g-j}\cdot  {\mathrm{Res}}^2_{\mathring{p}_{j}}\dfrac{dx \wedge dy}{x^{2g}yf^{(j)}_{2g,1}}\nonumber\\
&=\dfrac{-1}{2\pi}\cdot (\mathring{a}_{j}\mathring{x}_{j}^{2(g-j+1)})\cdot\dfrac{1}{\ay \mathring{x}_{j}^{g-1}\sqrt{8(g-j+1)(g+1)}}\\
&=\dfrac{\sqrt{g+1}}{\pi j\sqrt{2(g-j+1)}}.\nonumber
\end{align}

For the periods of $\varpi_j$, we start with those of the regulator class.  Writing $\vf_j:=x^{j-1}F_{2g,1}^{\ua}(x,y)-a_j$,  (with the sign flip from our choice of $\gamma_j$) yields
\begin{align}\label{1.3}
  \dfrac{1}{2\pi \mathbf{i}}R_{\gamma_j}(\underline{a})&\underset{\mathbb{Q}(1)}{\equiv}\log (a_{j})-\sum_{m>0}\dfrac{(-a_{j})^{-m}}{m} [\varphi_j^m]_{\uo} \nonumber\\ 
          &=\log (a_{j})-\sum_{m>0}\dfrac{(-a_{j})^{-m}}{m} \times \\ &\textstyle\bm{[(}\underbrace{x^{j}}_{=:A_j}+\underbrace{x^{j-1}y}_{=:B_j}+\textstyle\sum_{\substack{k=1 \\ k\neq j}}^{g+1}a_{k}\underbrace{x^{j-k}}_{=:C^{k}_j}+\underbrace{x^{j-2g-1}y^{-1}}_{=:D_j}\bm{)}^m\bm{]}_{\uo}\nonumber
\end{align}
where $[L]_{\uo}$ stands for the constant term (in $x,y$) appearing in the Laurent polynomial $L$. Now, given $l_1,l_2,\cdots,l_g\in \mathbb{Z}$, we define
\begin{align}
    \mathfrak{l}_{j}&:=\dfrac{1}{j}\bigg((g+1)(2l_{j}+l_{g+1})+\sum\limits_{\substack{k=1 \\ k\neq j}}^{g}kl_{k}\bigg)\\
    \mathfrak{l'}_{j}&:=\dfrac{1}{j}\bigg((g-j+1)(2l_{j}+l_{g+1})+\sum\limits_{\substack{k=1 \\ k\neq j}}^{g}(k-j)l_{k}\bigg),\text{ and put}\\
    \mathcal{L}_{j}&:=\{(l_1,l_2,\cdots,l_g)\in \mathbb{Z}^g_{\geq 0}\mid\mathfrak{l}'_{j}\in \mathbb{Z}_{\geq 0}\}\setminus\{(0,\cdots,0)\}
\end{align}
Note that $\mathfrak{l}'_{j}\in\mathbb{Z}_{\geq 0}\implies \mathfrak{l}_{j}\in\mathbb{Z}_{\geq 0}$. The upshot of this construction is if $L_{j},L'_{j}\in\mathbb{Z}_{\geq 0}$ are such that
\begin{align}
A_{j}^{L_{j}}B_j^{L'_{j}}\prod\limits_{\substack{k=1 \\ k\neq j}}^{g+1}(C^{k}_{j})^{l_{k}}D_j^{l_{j}}&=1\text{ and}\\
L_{j}+L'_{j}+\sum\limits_{k=1}^{g+1}l_{k}&=m
\end{align}
then $L_{j}=\mathfrak{l}'_{j}$, $L'_j=l_j$ and $m=\mathfrak{l}_{j}$. Thus the lattice $\mathcal{L}_{j}\subset \mathbb{Z}^{g}$ encodes all possible constant terms appearing in \eqref{1.3}, giving
\begin{equation}
    \dfrac{1}{2\pi \mathbf{i}}R_{\gamma_{j}}(\underline{a})\underset{\mathbb{Q}(1)}{\equiv}\log(a_{j})~-\sum_{\mathcal{L}_{j}}\dfrac{\Gamma(\mathfrak{l}_{j})}{\Gamma(1+\mathfrak{l}'_{j})\Gamma^2(1+l_{j})\prod\limits_{\substack{k=1 \\ k\neq j}}^{g+1}\Gamma(1+l_{k})}\frac{(-a_{j})^{-\mathfrak{l}_{j}}}{\mathfrak{l}_j}\prod\limits_{\substack{k=1 \\ k\neq j}}^{g+1}a_{k}^{l_{k}}.\label{0.9}
\end{equation}
For the classical periods $\Pi_{j\ell}=\int_{\gamma_j}\varpi_{\ell}=\tfrac{1}{2\pi\ay}\delta_{a_{\ell}}R_{\gamma_j}$, it is clear from \eqref{0.9} that $\Pi_{j\ell}$ vanishes on the $a_j$-axis for $\ell\neq j$.  Focusing then on
\begin{equation}
\Pi_{jj}(\underline{a})=\int_{\gamma_j}\varpi_j=1+\sum_{\mathcal{L}_{j}}\dfrac{\Gamma(1+\mathfrak{l}_{j})}{\Gamma(1+\mathfrak{l}'_{j})\Gamma^2(1+l_{j})\prod\limits_{\substack{k=1 \\ k\neq j}}^{g+1}\Gamma(1+l_{k})}(-a_{j})^{-\mathfrak{l}_{j}}\prod\limits_{\substack{k=1 \\ k\neq j}}^{g+1}a_{k}^{l_{k}},
\end{equation}
we set $\underline{a}_i=0$ for $i\neq j,g+1$ to obtain 
\begin{align}
\mathcal{S}:=1+\sum_{\tiny{l_{j},l_{g+1}} \in \mathbb{Z}_{>0}}\dfrac{\Gamma(1+\tfrac{g+1}{j}l_{j})}{\Gamma(1+\tfrac{g+1}{j}l_{j})\Gamma^2(1+l_{j})\Gamma(1+l_{g+1})}(-{a}_{j})^{-\tfrac{g+1}{j}l_{j}}a_{g+1}^{l_{g+1}}.
\end{align}
Recall that $\kappa_{j}:=\mathrm{gcd}(j, g+1)$. Let us shift indices by renaming $l_n\to l_{g+1}+2l_{n}$ and define,
\begin{equation}
\begin{split}
n_{j}:&=\dfrac{j}{\kappa_{j}}, \mspace{50mu} m_{j}:=\dfrac{g+1}{\kappa_{j}}=\dfrac{(g+1)n_j}{j},\\
r_{j}:&=\dfrac{l_{j}}{n_{j}},\mspace{50mu}\text{and}\mspace{50mu}s_j:=a_j^{-m_j}.
\end{split}
\end{equation}
Clearly $n_j,m_j,r_j\in \mathbb{Z}_{>0}$. Now we have a power series of the form
\begin{equation}
    \mathcal{S}=1+\sum_{r_{j}\in \mathbb{N}}{\dfrac{(-1)^{m_jr_j}\Gamma(1+m_{j}r_{j})a_{g+1}^{n_{j}r_{j}-2l_{j}}}{\Gamma^2(1+\tfrac{m_{j}-n_{j}}{2}r_{j})\Gamma^2(1+l_j)\Gamma(1+n_{j}r_{j}-2l_{j})}} s_j^{r_{j}} =: \sum_{r_j} b_{r_j}s_j^{r_j}.
\end{equation}
Let $\mathring{s}_j:=\mathring{a}_j^{-m_j}$. Setting $a_{g+1}=2$ and applying Lemma \ref{lemma_01},
\begin{equation}\label{gouri}
    \dfrac{\Gamma(1+m_{j}r_{j})}{\Gamma^2(1+\frac{m_{j}-n_{j}}{2}r_{j})\Gamma(1+n_{j}r_{j})} \approx \dfrac{(-1)^{m_jr_j}\sqrt{2m_{j}}}{2\pi r_{j}n_{j}\sqrt{m_{j}-n_{j}}}\mathring{s}_j^{r_{j}}
\end{equation}
from which we may conclude that
\begin{equation}
    \lim_{r_{j}\to \infty}b_{r_{j}}\cdot r_{j}\cdot \mathring{s}_j^{-r_{j}}=\dfrac{\sqrt{2m_{j}}}{2\pi n_{j}\sqrt{m_{j}-n_{j}}}.
\end{equation}
Observing that
\begin{align}
\mathrm{Res}_{\mathring{p}_{j}}\varpi_{j}=\dfrac{\sqrt{g+1}}{\pi j\sqrt{2(g-j+1)}}=\dfrac{\sqrt{2m_{j}}}{2\pi j\sqrt{m_{j}-n_{j}}},
\end{align}
we apply (\ref{0.7}) to obtain
\begin{equation}
\bm{\kappa}_{jj}=\dfrac{\underset{r_j\to\infty}{\lim}b_{r_{j}}\cdot r_j\cdot \mathring{s}_j^{r_j}}{\mathrm{Res}_{\mathring{p}_{j}}\varpi_{j}}=\dfrac{j}{n_{j}}=\kappa_{j}.
\end{equation}
This concludes the proof of Theorem \ref{thm_02}.

\begin{rem}
Notice that $\kappa_1=\kappa_g=1$. We document $\underline{\kappa}:=(\kappa_1,\ldots,\kappa_n)$ for $g=2,\dots,10$ in Table \ref{tab_1}.
\begin{table}[h!]
    \centering
    \begin{tabular}{|c|c |}
\hline
      $g$   & $\underline{\kappa}$  \\
      \hline
     2    & (1,1)\\
     \hline
     3    & (1,2,1)\\
     \hline
     4    & (1,1,1,1)\\
     \hline
     5    & (1,2,3,2,1)\\
         \hline
     6    & (1,1,1,1,1,1)\\
          \hline
     7    & (1,2,1,4,1,2,1)\\
          \hline
     8    & (1,1,3,1,1,3,1,1)\\
          \hline
     9   &  (1,2,1,2,5,2,1,2,1)\\
          \hline
     10   & (1,1,1,1,1,1,1,1,1,1)\\    
          \hline 
    \end{tabular}
   \captionsetup{justification=centering}
    \caption{Conifold multiples for small genera.}
    \label{tab_1}
\end{table}
\end{rem}

\subsection{Normalization of the conifold fibers}\label{S4c}

Recall that for the family $\C^{\ua}_{m,n}$ attached to $F^{\underline{a}}_{m,n}$, the \textit{maximal conifold} point $\underline{\hat{a}}\in(\CC^*)^{g}$ is defined to be the unique point (if it exists) on the boundary of the region of convergence of the series \eqref{0.9} where $\mathcal{C}_{m,n}^{\hat{\ua}}$ acquires $g$ nodes (labeled by $\hat{p}_j:=(\hat{x}_j,\hat{y}_j)$). 

\begin{rem}
Note that uniqueness and existence of $\hat{\ua}$ were essentially proven in \cite[\textit{Remark} 4.11]{DKSB}. We provide a different, albeit short explanation - the lattice triangle $\Delta_{m,n}$ gives rise to a unique, isoradial dimer model\footnote{ See \cite[\textbf{Theorem 1.2}]{UY}.}. This phenomenon will be discussed in full generality in an upcoming work.
\end{rem}

We furnish a concrete example locating $\hat{\ua}$ in low-genera case.

\begin{example}
Consider the (untempered) family $\mathcal{C}_{4,1}$ corresponding to a local $\mathbb{C}^3/\mathbb{Z}_6$ geometry cut out by
\begin{equation}
    F_{4,1}(x,y)=x+y+a_1+a_2x^{-1}+a_3x^{-2}+x^{-4}y^{-1}=0.
\end{equation}
$a_3$ is a mass parameter, and
\begin{equation}
   z_1=\dfrac{a_1a_3}{a_2^2},~~z_2=\dfrac{a_2}{a_1^2},~~z_3=\dfrac{1}{a_3^2}
\end{equation}
are the large complex structure parameters. The discriminant locus is obtained by setting
\begin{align}
 &729z_1^4(1-4z_3)^2z_2^4+108z_1^3(9z_2-2)(4z_3-1)z_2^2+(1-4z_2)^2\nonumber\\
 &-4z_1(36z_2^2-17z_2+2)+2z_1^2(108(4z_3+1)z_2^3-27(28z_3-5)z_2^2)\\
 &+72((4z_3-1)z_2-32z_3+8)=0.\nonumber
\end{align}
Temperedness amounts to letting $z_3=\dfrac{1}{4}$, and the maximal conifold point 
\begin{equation}
\underline{\hat{z}}=\bigg(\dfrac{4}{27},\dfrac{1}{4}\bigg)
\end{equation} 
can once again be recovered from transverse intersection.
\end{example}

The \textit{ansatz} that allows us to bypass such considerations is described by
\begin{prop}\label{prop_04}
Let $\mathcal{T}_m$ denote the $m^{\text{th}}$ \textup{Chebyshev polynomial} of the first kind; this is a degree-$m$ polynomial characterized by $\mathcal{T}_{m}(\cos \theta)=\cos m\theta$. Then we have
\begin{equation}\label{ep04}
    F^{\underline{\hat{a}}}_{g,g}(x,(-1)^jx^{-g})=\dfrac{2}{x^{g+1}}(\T_{2g+2}(\tfrac{\sqrt{x}}{2})+(-1)^j).
\end{equation}
It follows that
\begin{align}
\hat{a}_{j}&=(-1)^{j}\dfrac{(2g+2)(2g-j+1)!}{j!(2g-2j+2)!}\;\;\;\text{and}\label{ep04a} \\
\hat{x}_{j}&=(-1)^{j/g}\hat{y}_{j}^{-1/g}=(-1)^{j} 4\cos^2\bigg(\dfrac{\pi j}{2g+2}\bigg)
\end{align}
for $j=1,\ldots,g$.  In particular, $\underline{\hat{a}}\in\mathbb{Z}^g$.
\end{prop}

\begin{proof}
That $\hat{x}_j\in \mathbf{Z}(\text{RHS}\eqref{ep04})$ is immediate from the defining property of $\T_{2g+1}$, and the $\hat{x}_j$ are distinct and different from $4$.  Moreover, writing $\mathcal{U}_m$ for the $m^{\text{th}}$ Chebyshev polynomial of the second kind, the relation $(\T_{2g+2}(w)-1)(\T_{2g+2}(w)+1)=(w^2-1)U_{2g+1}(w)$ guarantees that all roots other than $4$ of $(\T_{2g+2}(\tfrac{1}{2x})+1)$ have even multiplicity.  So they all have multiplicity $2$ and are precisely the $\{\hat{x}_j\}$.

The polynomial $\hat{F}(x,y):=x+y+\sum_{j=1}^g \hat{a}_j x^{1-j}+a_{g+1}x^{-g}+x^{-2g}y^{-1}$, with $\hat{a}_j$ as in \eqref{ep04a}, satisfies $\hat{F}(x,(-1)^jx^{-g})=\text{RHS}\eqref{ep04}$ by standard results on coefficients of $\T_m$.  Clearly $\hat{F}(\hat{p}_j)=0$, and the $\{\hat{p}_j\}$ are in fact singularities of $\mathbf{Z}(\hat{F})$ since $\tfrac{\partial \hat{F}}{\partial x}(x,(-1)^jx^{-g})=2\tfrac{d}{dx}(\hat{F}(x,(-1)^jx^{-g}))$ and they are double roots of $\hat{F}(x,(-1)^jx^{-g})$.  Therefore, by Proposition \ref{Proposition_0.2}, they are all nodes.  Since one can also check that \eqref{0.9} converges at $\hat{p}_j$, $\mathbf{Z}(\hat{F})$ is the maximal conifold curve.
\end{proof}

\begin{rem}
Of course, Proposition \ref{prop_04} recovers the predicted maximal conifold point for the $g=2$ family $\mathcal{C}_{4,1}$, namely $\hat{a}_1=-6,\hat{a}_2=9$.  Table \ref{tab_2} gathers $\mathcal{T}_{2g+2}$ and $\underline{\hat{a}}$ for a few low genus cases.
\end{rem}
\begin{table}[h!]
    \centering
    \begin{tabular}{|c|c |c|}
\hline
      $g$ &$\mathcal{T}_{2g+2}(x)$  & $\underline{\hat{a}}$  \\
      \hline
     1    & $8x^4-8x^2$+1 & -4 \\
     \hline
     2    & $32x^6-48x^4+18x^2-1$ &(-6,9)\\
     \hline
     3    & $128x^8-256x^6+160x^4-32x^2+1$ & (-8,20,-16)\\
     \hline
     4    & $512x^{10}-1280x^8+1120x^6-400x^4+50x^2-1$ &(-10,35,-50,25)\\
         \hline
    \end{tabular}\normalsize
    \caption{Maximal conifold points for low genera.}
    \label{tab_2}
\end{table}
$\hat{\mathcal{C}}_{2g,1}$ admits  $g$ distinct uniformizations by $\mathbb{P}^1$(proven in appendix), given by $z\mapsto (\hat{X}_{j}(z),\hat{Y}_{j}(z))$, with
\begin{align}\label{4.57}
 \hat{X}_{j}(z)&=\frac{\hat{x}_{j}\left(1-\tfrac{\zeta_{2g+2}^j}{z}\right)^{2}}{\left(1-\tfrac{{\zeta^{2j}_{2g+2}}}{z}\right)\left(1-\tfrac{1}{z}\right)}\;\;\;\text{and}\\ \label{4.58}
 \hat{Y}_{j}(z)&=\frac{\hat{y}_j\left(1-z\right)^{2g+1}}{\left(1-\tfrac{z}{\zeta^{j}_{2g+2}}\right)^{2g}\left(1-\tfrac{z}{\zeta_{2g+2}^{2j}}\right)},
\end{align}
all of whom map $z=0,\infty$ to $\hat{p}_{j}$, implying that the path joining $z=0$ to $z=\infty$ on $\PP^1$ is sent (by the $j^{\text{th}}$ map) to $\hat{\gamma}_{j}$.  As per\cite[\S6.2]{DK}, a formal divisor $\hat{\mathcal{N}}_j$ on $\PP^1\setminus\{0,\infty\}$ is introduced to each uniformization:  for $X(z)=c_1\prod_j(1-\tfrac{\alpha_{j}}{z})^{d_{j}}$ and $Y(z)=c_2\prod_k (1-\tfrac{z}{\beta_{k}})^{e_k}$, this divisor is $\mathcal{N}:=\sum_{j,k}d_j e_k [\tfrac{\alpha_j}{\beta_k}]$.  According to [loc. cit.], the imaginary part of $\int_0^{\infty} R\{X(z),Y(z)\}$ is then given by $D_2(\mathcal{N}):=\sum_{j,k}d_je_k D_2(\tfrac{\alpha_j}{\beta_k})$.

Working in $\mathcal{B}_2(\mathbb{C})$, we invoke scissors congruence relations
\begin{align}
[\xi]+[\tfrac{1}{\xi}]=0,~ [\xi]+[\overline{\xi}]=0, ~[\xi]+[1-\xi]&=0 ~\text{and}\\
[\xi_1]+[\xi_2]+[\tfrac{1-\xi_1}{1-\xi_1\xi_2}]+[\tfrac{1-\xi_2}{1-\xi_1\xi_2}]+[1-\xi_1\xi_2]&=0
\end{align} 
to obtain,
\begin{align}
\hat{\mathcal{N}}_{j}=2(2g+2)[1+\zeta^{j}_{2g+2}].\nonumber
\end{align}
Consequently we have the identity
\begin{align}\label{0.67}
 D_2(\hat{\mathcal{N}}_{j})&=2(2g+2)D_2(1+\zeta^{j}_{2g+2}).
\end{align}
We are now ready to prove Theorem \ref{thm_01}.  By the previously mentioned result of \cite[\S6.2]{DK}, we know that $\Im(R_{\hat{\gamma}_j})=D_2(\hat{\mathcal{N}}_j)$ or
\begin{equation}
    \Re(\tfrac{1}{2\pi\ay}R_{\hat{\gamma}_j})=\tfrac{1}{2\pi}D_2(\hat{\mathcal{N}}_j).
\end{equation}
Next, Proposition \ref{thm_02} tells us that $R_{\gamma_j}(\hat{\ua})=\kappa_j R_{\hat{\gamma}_j}$, while \eqref{ep04a} and \eqref{0.9} ensure that (mod $\QQ(1)$) $\tfrac{1}{2\pi\ay}R_{\gamma_j}(\hat{\ua})$ hence $\tfrac{1}{2\pi\ay}R_{\hat{\gamma}_j}$ is real.  Combining this with (\ref{0.67}) gives
\begin{align}
\dfrac{1}{2\pi \mathbf{i}}R_{\gamma_{j}}(\underline{\hat{a}})=\dfrac{1}{2\pi \mathbf{i}}\kappa_j R_{\hat{\gamma}_{j}}\underset{\mathbb{Q}(1)}{\equiv}\dfrac{(2g+2)\kappa_j}{\pi}D_2(1+\zeta^{j}_{2g+2}),\label{0.15}
\end{align}
whence (\ref{0.3}) follows by setting $j=1$ in  (\ref{0.15}).

\section{Explicit series identities}\label{S4d}
Any torsion modulo $\mathbb{Q}(1)$ in (\ref{0.15}) is eliminated with regards to (\ref{0.9}) as both sides are real,\footnote{after changing $\log(\hat{a}_j)$ to $\log(|\hat{a}_j|)$} and brings forth the relationship 
\begin{align}
    \dfrac{(2g+2)\cdot\mathrm{gcd}(j,g+1)}{\pi}&D_2(1+\zeta^{j}_{2g+2})=\log(|\hat{a}_{j}|)\label{1.11}-\nonumber\\
    &\sum_{\mathcal{L}_{j}}\dfrac{\Gamma(\mathfrak{l}_{j})}{\Gamma(1+\mathfrak{l}'_{j})\Gamma^2(1+l_{j})\prod\limits_{\substack{k=1 \\ k\neq j}}^{g+1}\Gamma(1+l_{k})}\frac{(-\hat{a}_{j})^{-\mathfrak{l}_{j}}}{\mathfrak{l}_j}\sum\limits_{\substack{k=1 \\ k\neq j}}^{g+1}\hat{a}_{k}^{l_{k}}
\end{align}
valid for $j=1,\dots,g$. For the family $\mathcal{C}_{4,1}$, Table \ref{tab_1} and Table \ref{tab_2} say that $\underline{\kappa}=(1,1)$ and $\underline{\hat{a}}=(-6,9)$, thus proving,
\begin{align}
\dfrac{6}{\pi}D_2(1+ e^{\pi \ay/6})&=\log 6 - \mspace{-15mu}\sum_{l_1,l_2,l_3\in \mathbb{Z}_{\geq 0}}{}^{\mspace{-25mu}\prime}\dfrac{\Gamma(6l_1+2l_2+3l_3)(-6)^{-6l_1-2l_2-3l_3}9^{l_2}2^{l_3}}{\Gamma(1+4l_1+2l_3)\Gamma^2(1+l_1)\Gamma(1+l_2)\Gamma(1+l_3)} \nonumber \\
&=\log 6 - \mspace{-15mu} \sum_{m,r,s\in \mathbb{Z}_{\geq 0}}{}^{\mspace{-25mu}\prime}\dfrac{(-1)^{s}\Gamma(6m+2r+3s)3^{-6m-3s}2^{s}}{\Gamma(1+4m+r+2s)\Gamma^2(1+m)\Gamma(1+r)\Gamma(1+s)},\nonumber
\end{align}
an identity conjectured in \cite{CGM}.

\appendix

\section{Proof of Parameterizations}

We tie up the remaining loose end, namely proof of parameterizations \ref{4.57}, \ref{4.58} by slightly altering the proof presented in \cite[\S4.3]{DKSB} - that these maps are of degree 1 is straightforward from considering the intersection of the image and the boundary divisors of $\mathbb{P}_{\Delta_{2g,1}}$. Now, $F^{\ua}_{2g,1}$ is irreducible - this is also easy, a rather direct application of \cite[\textbf{Prop. 2.6}]{Gao}. The difficult part is to show that the image is contained in $\C^{\hat{\ua}}_{2g,1}$. We begin by observing that\footnote{This factorization is due to the dimer model phenomenon, and will be explained in an upcoming work.},

\begin{align*}
F^{\hat{\ua}}_{2g,1}(x,y)&=\prod _{i=-\frac{2g+1}{2}}^{\frac{2g+1}{2}} \bigg(1+\zeta_{2g+2}^{2gi}x^{\frac{2g}{2g+2}}y^{\frac{2}{2g+2}}+\zeta_{2g+2}^{-i}x^{\frac{2g+1}{2g+2}}y^{\frac{1}{2g+2}}\bigg) \\
&:=\prod _{i=-\frac{2g+1}{2}}^{\frac{2g+1}{2}} P_i(x,y)
\end{align*}

Let $(\hat{x},\hat{y})$ be the node of $F^{\hat{\ua}}_{2g,1}(x,y)$ corresponding to $i=\pm\frac{2g+1}{2}$, then the map from $\phi:\mathbb{P}^1_z\to (\CC^*)^2_{X,Y}$ parametrizing $F^{\hat{\ua}}_{2g,1}(x,y)$ such that the fiber over $(\hat{x},\hat{y})$ is $\{0,\infty\}$ can be written as follows,

\begin{align*}
   X&=\hat{x}\dfrac{(z-1)^{2}}{(z-\zeta_{2g+2}^{2g+1})(z-\zeta_{2g+2})}\\
Y&=\hat{y}\dfrac{
(z-\zeta_{2g+2}^{2g+1})^{2g+1}}{(z-1)^{2g}(z-\zeta_{2g+2})}
\end{align*}

\begin{lem}\label{lemma_01}
 $$a=\dfrac{\alpha^{n+1}}{\beta^n \theta},~~b=\dfrac{\beta^{m+1}}{\alpha^m \theta}$$
Then the following identity holds, $$a^{\frac{m+1}{N}}b^{\frac{n}{N}}=\dfrac{\alpha}{\theta},~~a^{\frac{m}{N}}b^{\frac{n+1}{N}}=\dfrac{\beta}{\theta}$$
\end{lem}

\begin{proof} Clearly we have $$a^{m+1}b^{n}=\alpha^{(m+1)(n+1)-mn}\beta^{-n(m+1)+n(m+1)}\theta^{-m-1-n}=\dfrac{\alpha^N}{\theta^N}$$
as was claimed. The other identity follows the same way.
\end{proof}
\newpage

\begin{prop}\label{Proposition_0.2} $\phi$ is indeed the desired parametrization, that is,
$$F^{\hat{\ua}}_{2g,1}(X,Y)=0$$

\begin{proof} Using \textbf{Lemma \ref{lemma_01}} with  $$m=2g,~n=1,~\alpha = z-1,~ \beta= z-\zeta_{2g+2}^{2g+1}, ~\theta = z-\zeta_{2g+2} $$ we have, 
\begin{align*}
 F^{\hat{\ua}}_{2g,1}(X,Y)&=\prod _{i=-\frac{2g+1}{2}}^{\frac{2g+1}{2}} \bigg(1+\zeta_{2g+2}^{2gi}\hat{x}^{\frac{2g}{2g+2}}\hat{y}^{\frac{2}{2g+2}}\dfrac{z-\zeta_{2g+2}^{2g+1}}{z-\zeta_{2g+2}}+\zeta_{2g+2}^{-i}\hat{x}^{\frac{2g+1}{2g+2}}\hat{y}^{\frac{1}{2g+2}}\dfrac{z-1}{z-\zeta_{2g+2}}\bigg)\\
 &=\prod _{i=-\frac{2g+1}{2}}^{\frac{2g+1}{2}} (z-\zeta_{2g+2})\bigg(zP_i(\hat{x},\hat{y})-\zeta_{2g+2} P_{i+1}(\hat{x},\hat{y})\bigg)\\
\end{align*}
However as
$$P_{\frac{2g+1}{2}+1}(x,y)=P_{\frac{2g+3}{2}}(x,y)=P_{-\frac{2g+1}{2}}(x,y)$$
and by assumption $$P_{\frac{N-1}{2}}(\hat{x},\hat{y})=P_{-\frac{N-1}{2}}(\hat{x},\hat{y})=0$$
we can conclude that $$F^{\hat{\ua}}_{2g,1}(X,Y)=0$$

\end{proof}

\end{prop}

\end{document}